\documentclass[twocolumn,letterpaper]{IEEEAerospaceCLS}  

\usepackage{subcaption}
\usepackage[]{graphicx}    
\usepackage[colorlinks=true, urlcolor=blue, linkcolor=black, citecolor=black]{hyperref} 
\usepackage{amsmath} 
\usepackage{mathtools}
\usepackage{tabularx}
\usepackage{amsfonts}
\usepackage{bm}
\usepackage{cite}
\usepackage{xurl}
\usepackage{svg}

\newtheorem{theorem}{Theorem}[section]
\newtheorem{lemma}[theorem]{Lemma}

\newtheorem{corollary}{Corollary}

\newcommand{\ignore}[1]{}  
\begin{document}
\title{Probabilistic Connectivity Analysis of Recursive Satellite Release for Formation Initialization}

\author{%
Hideki Yoshikado\\ 
Interstellar Technologies Inc.,\\
Hiroo, Hokkaido 089-2113, Japan\\
hideki.yoshikado@istellartech.com
\and 
Yuta Takahashi\\
Interstellar Technologies Inc.,\\
Hiroo, Hokkaido 089-2113, Japan\\
stateofyuta@gmail.com
\thanks{\footnotesize 979-8-3315-7360-7/26/$\$31.00$ \copyright2026 IEEE}              
}

\maketitle

\thispagestyle{plain}
\pagestyle{plain}

\maketitle

\thispagestyle{plain}
\pagestyle{plain}

\begin{abstract}
In the initial deployment of large-scale distributed space systems using small satellites, achieving a reliable transition to passively stable orbits while maintaining inter-satellite distances within effective control and communication ranges is crucial, particularly given the presence of deployment errors and uncontrolled coasting phases. This study presents a framework for designing formation initialization that provides probabilistic safety guarantees. The scope covers the initial deployment phase—from sequential release by a single carrier to commissioning, control activation, and transition to passive stabilization. Strict separation limits during initialization necessitate low release velocities to minimize relative drift before control activation. However, in the low-velocity regime, the allowable tolerances for release velocity and angular rate errors tighten significantly to satisfy distance constraints, making hardware requirements a critical bottleneck. To address this, we model the initialization sequence as a stochastic process and derive closed-form constraints on deployment errors and control activation intervals. These conditions ensure that inter-satellite distances remain within the allowable separation limit with a prescribed probability. Monte Carlo simulations, configured using the error bounds and intervals derived from the proposed constraints, demonstrate that inter-satellite distances are successfully maintained within the allowable range. The proposed framework enables the safe initialization of large-scale distributed space systems by translating strict separation constraints into quantifiable hardware requirements.
\end{abstract} 
\tableofcontents
\section{Introduction}
\begin{figure}[tb!]
\centering
  \begin{minipage}[b]{0.495\columnwidth}
    \centering
    \includegraphics[width=1\columnwidth]{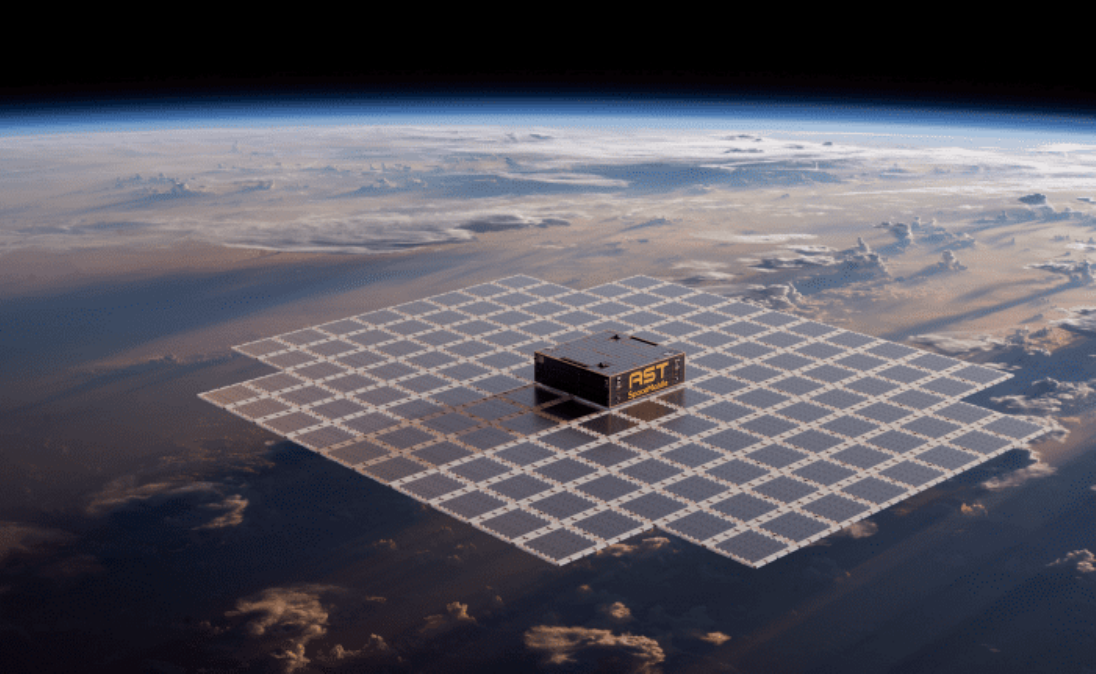}
    \subcaption{Monolithic aperture antenna.}
  \end{minipage}
  \begin{minipage}[b]{0.494\columnwidth}
    \centering
\includegraphics[width=1\columnwidth]{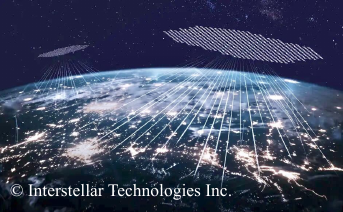}
    \subcaption{Distributed aperture antenna.}
  \end{minipage}\\
\begin{minipage}[b]{1\columnwidth}
    \centering
\includegraphics[width=1.0\linewidth]{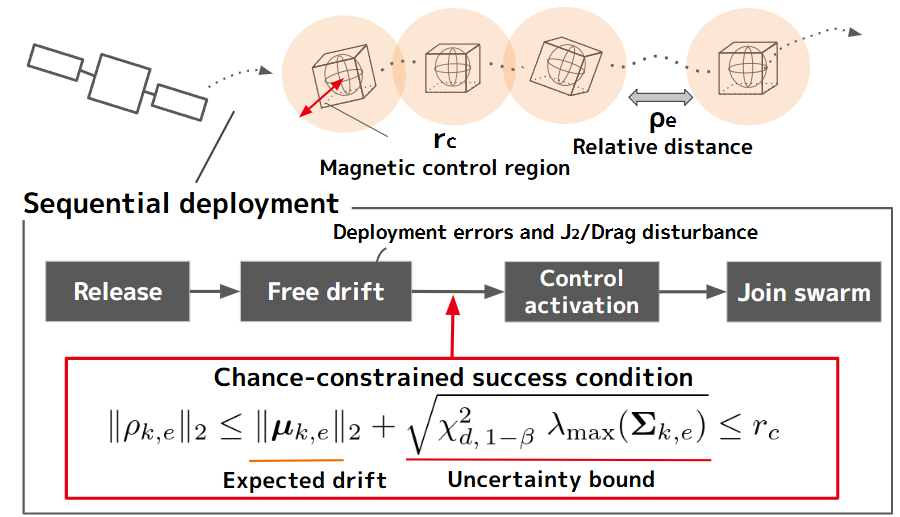}    
    \subcaption{Distance-constraint formation initialization framework.}
  \end{minipage}
  \caption{Monolithic space antenna array (Conceptual illustrations of the BlueWalker3 satellite. © AST SpaceMobile), distributed space antenna array \cite{takahashi2025distance,shim2025feasibility,takahashi2025scalable,shim2026early,morioka2024dense}, and our fuel-free formation initialization framework. The chance-constrained success condition guarantees that the relative distance remains within the magnetic control region under release errors, $J_2$ effect, and atmospheric drag.}
   \label{fig:graphical_abstruct}
\end{figure}

Deploying resource-constrained small-satellite swarms into a distributed space system with high reliability can support broader space utilization. Driven by growing demand for high-capacity communications and advanced observation, satellite swarms have been studied as an alternative to a single large spacecraft \cite{Guo2021,QADIR2023296,Radhakrishnan2016,Tuzi2023,Kong2004,Hadaegh2016,She2019,takahashi2025distance,shim2025feasibility,takahashi2025scalable,shim2026early,morioka2024dense}. By coordinating multiple small spacecraft, these systems can achieve large-aperture-equivalent mission objectives with reduced mass while providing adaptability and redundancy \cite{Zhu2022,Geraci2023,Saleh2003}. Accordingly, small-satellite-based distributed systems may lower cost and risk for future missions and infrastructure. In particular, safety during swarm initialization becomes a primary bottleneck, since loss of controllability or connectivity at this stage can prevent successful formation establishment.

Safe initialization of a small-satellite swarm requires that both controllability and inter-satellite communication constraints be satisfied throughout the early phase. In formation flight control for small satellites, propellantless approaches such as differential aerodynamic drag \cite{mazal2015rendezvous,BRUINSMA2023} and magnetic interactions \cite{shim2025feasibility,schweighart2005emff,youngquist2013alternating,takahashi2024neural,zhang2016angular,Foust2018,takahashi2022kinematics,takahashi2025coil,takahashi2025noda_mmh} have been widely investigated to maintain formations under limited onboard resources. Because these approaches provide limited control authority, the practical controllable range is governed by the distance scale \cite{takahashi2026graph}. Within that distance, the achievable control effect can balance the control effort required under disturbances. In addition, ground-based individual supervision becomes impractical as the swarm size increases, and inter-satellite link conditions impose further range constraints. To satisfy these control and communication ranges in the initial phase, low-speed release is required so that relative separation does not grow during the uncontrolled interval before control activation. Moreover, in the low-speed regime, release-induced rotation due to an offset line of action becomes non-negligible. Therefore, a unified framework is required to handle both release-velocity and rotation-induced errors while maintaining the distance-limited controllable region during the initialization sequence.

However, when low-speed release is assumed, the tolerable error budget consistent with the distance constraint can decrease rapidly. In addition to release velocity errors, tip-off rotation induced by a line-of-action offset can become a dominant contributor. For small satellites with a large area-to-mass ratio, this rotation can change attitude and effective cross-sectional area, which may promote growth of relative motion. Therefore, a key problem is to treat velocity errors and rotation-induced errors jointly and to maximize the allowable error bounds at release within a short deployment interval.

Evaluating hardware requirements for multi-satellite deployment calls for a deployment logic that consistently represents the initialization sequence. The sequence includes release, an uncontrolled commissioning interval, activation of active control, and a transition to passively stable motion, whereas prior studies often treated separation design and stabilization analysis independently \cite{Jeffrey2003,inamori2022inorbit}. Such separation can hinder scalable guarantees under dispersed release conditions because the stochastic state dimension increases rapidly with swarm size. This study models sequential releases recursively and derives allowable error bounds that keep the swarm within a permissible separation distance. The resulting framework can enable scalable hardware requirement design while maintaining probabilistic guarantees for distributed space systems.

\subsection{Contributions and Paper Organization}

This study aims to establish practical design guidelines for safe initial deployment of resource-constrained satellite swarms, so that a distributed space system can be realized with quantified safety margins under limited control authority. Two main contributions are provided in support of this objective. First, chance-constrained conditions are derived to translate limited control authority, which is particularly relevant during passive stabilization, into hardware-level requirements for the deployment mechanism, including allowable tolerances in release velocity, release direction, tip-off angle, and induced angular-rate errors, as well as permissible deployment intervals. Second, numerical simulations indicate that, in sequential deployment design for distributed space systems with massive swarms, the derived deployment requirements converge beyond a certain swarm size. This behavior suggests that system-level safety can be assessed through evaluations on a local subset of satellites, enabling a scalable initialization design despite limited onboard resources.

The remainder of this paper is organized as follows. Section~2 summarizes the fundamental models and assumptions used in this study. Section~3 formulates the sequential deployment and activation protocol as a mathematical problem and introduces the safety criterion under limited control authority. Based on this formulation, Section~4 derives scalable chance-constraint conditions by modeling error propagation through the deployment sequence and relating it to swarm consensus control. Section~5 presents numerical case studies and Monte Carlo analyses to quantify allowable deployment errors and to illustrate the convergence property with increasing swarm size. Section~6 concludes by summarizing the main findings and implications for future distributed space systems. To validate the proposed strategy, we assume full knowledge of neighboring satellite states through established state estimation techniques, such as range-based swarm navigation \cite{timmons2022range}, magnetic-field–based relative navigation with sequential filtering \cite{shibata2024relative}, and MTQ-enabled femtosatellite navigation systems \cite{hu2019development}.
\section{Preliminaries}
This section summarizes the mathematical framework used to describe the initialization sequence of the satellite swarm considered in this study. Specifically, it introduces the J2-averaged Hill equations for modeling relative motion in Earth orbit, algebraic graph theory for representing inter-satellite interactions, consensus control for passive stabilization of the swarm, and a Gaussian-distribution-based formulation for probabilistic safety assessment.

\subsection{Relative Dynamics with Averaged $J_2$ Gravity Effects}
In this subsection, intersatellite relative motion is modeled using the Hill--Clohessy--Wiltshire (HCW) framework, extended to include the mean effects of Earth’s oblateness as in \cite{SchweighartSedwick2002, takahashi2025distance,takahashi2025scalable}.
The classical HCW model assumes a circular reference orbit; in contrast, this study adopts a modified reference orbit that incorporates the $J_2$-averaged perturbation.
The resulting dynamics are obtained by linearizing about a virtual center associated with the averaged motion, while residual perturbations such as differential drag and higher-order gravitational terms are neglected \cite{takahashi2025distance,takahashi2025scalable}.

Given the initial relative state at $t=0$, the analytical solution for the relative position vector $\bm{r}(t) = [x(t), y(t), z(t)]^\top$ in the curvilinear Local-Vertical Local-Horizontal (LVLH) frame is expressed as the sum of a secular drift center $\bm{r}_o(t)$ and bounded periodic oscillations \cite{takahashi2025distance,takahashi2025scalable}:
\begin{equation}
\label{eq:analytical_solution}
\bm{r}(t) = 
\underbrace{
\begin{bmatrix} 2C_1 \\ C_4 - \epsilon_2 C_1 t \\ 0 \end{bmatrix}
}_{\text{Drift Center } \bm{r}_o(t)}
+
\underbrace{
\begin{bmatrix}
\frac{1}{c_+} r_{xy} \sin(\omega_{xy}t + \theta_{xy}) \\
\frac{2}{c_-} r_{xy} \cos(\omega_{xy}t + \theta_{xy}) \\
(r_z + lt) \sin(\omega_z t + \theta_z)
\end{bmatrix}
}_{\text{Bounded Oscillation}}.
\end{equation}
where $\omega_{xy}$ and $\omega_z$ are the in-plane and out-of-plane eigenfrequencies adjusted for $J_2$ effects, and $c_\pm$ and $\epsilon_2$ are coefficients depending on the orbital altitude and inclination. The geometry of the trajectory is parameterized by six integration constants $C_1, \dots, C_6$, which correspond to the relative orbital elements and describe satellite swarm trajectories as shown in Fig.~\ref{fig:grid_formation}.

The coefficients in Eq.~\eqref{eq:analytical_solution} are $c_{\pm}$ and $\epsilon_2$, and the in-plane eigenfrequency is $\omega_{xy}$. They are defined by \cite{takahashi2025distance,takahashi2025scalable}
$$
\begin{aligned}
  c_{\pm} &= \sqrt{1 \pm s_{J_2}}, \quad
  \epsilon_2 = \frac{3+5s_{J_2}}{c_+ c_-}\,\omega_{xy}, \quad
  \omega_{xy} = c_{-}\,\sqrt{\frac{\mu_g}{r_{\mathrm{ref}}^3}}, \notag\label{eq:coeff_main}\\
  s_{J_2} &= \frac{k_{J_2}\left(1+3\cos 2i_{\mathrm{ref}}\right)}{4\,\mu_g\,r_{\mathrm{ref}}^{2}}. \label{eq:sj2_def}
\end{aligned}
$$
Here, $s_{J_2}$ is a dimensionless parameter representing the effect of the averaged $J_2$ perturbation, $i_{\mathrm{ref}}$ is the inclination of the reference orbit, $r_{\mathrm{ref}}$ is the reference orbital radius, $\mu_g$ is the gravitational parameter of the central body, and $k_{J_2}$ is a constant that scales the $J_2$ contribution. 
These definitions specify $c_{\pm}$, $\epsilon_2$, and $\omega_{xy}$ without introducing undefined symbols.

The integration constants $C_1$--$C_6$ are determined from the initial state at $t=0$ expressed in the scaled in-plane coordinates $\bar{x}$ and $\bar{y}$. 
The scaled coordinates are defined by
\begin{equation*}
  \bar{x} = c_+ x, \qquad \bar{y} = c_- y,
\end{equation*}
and the corresponding scaled velocities are $\dot{\bar{x}} = c_+ \dot{x}$ and $\dot{\bar{y}} = c_- \dot{y}$.
Using the initial values $(\bar{x}, \bar{y}, z, \dot{\bar{x}}, \dot{\bar{y}}, \dot{z})$ at $t=0$, the constants are computed as
\begin{align}
  C_1 &= \frac{c_+}{c_-^{\,2}}\left(2\bar{x} + \frac{\dot{\bar{y}}}{\omega_{xy}}\right), &
  C_4 &= \frac{\,\bar{y} - 2\dot{\bar{x}}/\omega_{xy}\,}{c_-}, \notag\\
  C_2 &= \frac{\bar{y} - c_- C_4}{2}, &
  C_3 &= \bar{x} - 2c_+ C_1.
\end{align}
This ordering first defines the variables appearing in the expressions for $C_1$--$C_4$ and then uses them to compute the integration constants.


Since $C_1, C_4$ is uniquely determined by a linear combination of the initial relative position and velocity (specifically, the semi-major axis difference) at the moment of release \cite{takahashi2025distance,takahashi2025scalable}, any deployment error that results in $C_1 \neq 0$ will inevitably cause a secular drift, separating the satellites over time.
Therefore, ensuring $C_1 \approx 0$ is the primary objective for formation initialization, while $C_4$ determines the constant offset in the along-track direction.
\begin{figure}
    \centering
    \includegraphics[width=0.8\linewidth]{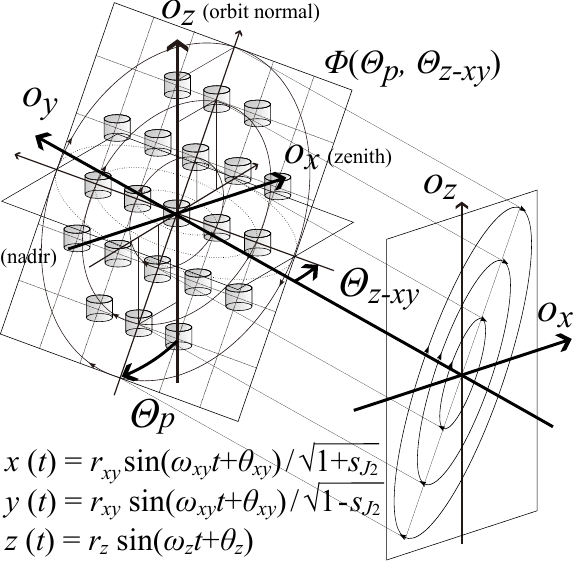}
    \caption{The grid formation of satellite swarm in a coplanar equidistant plane\cite{takahashi2025distance,takahashi2025scalable}.}
    \label{fig:grid_formation}
\end{figure}

\subsection{Algebraic Graph Theory}
\label{subsec:algraph}

This subsection summarizes the graph-theoretic notation used to model a satellite cluster and its expansion by sequential deployment.
A (simple, undirected) graph at stage $k$ is denoted by $\mathcal{G}_k=(\mathcal{V}_k,\mathcal{E}_k)$, where $\mathcal{V}_k$ is the node set and $\mathcal{E}_k$ is the edge set.
We write $n_k\coloneqq|\mathcal{V}_k|$ and $m_k\coloneqq|\mathcal{E}_k|$, and index the nodes as $\mathcal{V}_k=\{1,\dots,n_k\}$.
Each edge $e\in\mathcal{E}_k$ is assigned an arbitrary orientation only for algebraic representation; the physical quantities of interest will be expressed through edge-wise relative coordinates and thus do not depend on the orientation.

Let $E_k\in\mathbb{R}^{n_k\times m_k}$ be the oriented incidence matrix of $\mathcal{G}_k$.
For an edge $e=(i,j)$ oriented from $i$ to $j$, the corresponding column of $E_k$ has $+1$ at row $i$, $-1$ at row $j$, and $0$ elsewhere.
Given node states $x_i\in\mathbb{R}^{d}$ and the stacked node state $x^{(k)}\in\mathbb{R}^{dn_k}$, the stacked edge differences are represented by
\begin{align}
x^{(k)} &\coloneqq [x_1^\top,\dots,x_{n_k}^\top]^\top \in \mathbb{R}^{dn_k}, \nonumber\\
\rho^{(k)} &\coloneqq (E_k^\top\otimes I_d)x^{(k)} \in \mathbb{R}^{dm_k}.
\label{eq:prelim_edge_stack_general}
\end{align}
Here, $\otimes$ denotes the Kronecker product and $I_d$ is the $d\times d$ identity matrix.
The vector $\rho^{(k)}$ stacks the edge-wise differences, which correspond to $x_i-x_j$ for all $(i,j)\in\mathcal{E}_k$ up to the arbitrary edge orientation.

The (node) Laplacian and the edge Laplacian are defined by
\begin{align}
L_k &\coloneqq E_k E_k^\top \in \mathbb{R}^{n_k\times n_k}, \nonumber\\
L_{e,k} &\coloneqq E_k^\top E_k \in \mathbb{R}^{m_k\times m_k}.
\label{eq:prelim_edge_laplacian}
\end{align}
Both $L_k$ and $L_{e,k}$ are positive semidefinite.
When $\mathcal{G}_k$ is connected, $L_k$ has a single zero eigenvalue associated with the consensus subspace.
Moreover, $L_{e,k}$ has rank $m_k-1$ under the same connectivity condition.
We denote by $L_{e,k}^\dagger$ the Moore--Penrose pseudoinverse of $L_{e,k}$.

\subsection{Distributed Drift-Center Consensus Control for Long-Term Bounded Relative Motion}

As indicated by the analytical solution \eqref{eq:analytical_solution}, the J2-averaged model represents the long-term relative motion as the superposition of (i) the drift-center evolution, which captures the slow drift of the relative-orbit center, and (ii) a bounded elliptical circulation about that center.
The drift-center coordinates provide a compact description of the LVLH center of the relative orbit and directly encode the long-term drift behavior.
Accordingly, maintaining a formation geometry over long durations can be interpreted as regulating the drift-center coordinates among satellites so that the cluster does not diverge due to orbital drift.
In the following, the orbital dynamics that map the applied control to the drift-center evolution are assumed to be available in the orbit model part (not included here), while the inter-satellite coupling is described by the graph-theoretic objects introduced in the previous subsection.

In this work, it is assumed that each satellite can estimate (or exchange) drift-center coordinates with its neighbors through inter-satellite links, and that a distributed feedback law can be applied over the active communication/control graph.
Let $r_{o,i}(t)\in\mathbb{R}^{d}$ denote the drift-center coordinate of satellite $i$, and consider the currently connected cluster represented by $\mathcal{G}_k=(\mathcal{V}_k,\mathcal{E}_k)$.
Define the stacked drift-center state and its stacked edge-wise differences by
\begin{align}\label{eq:stack_ro}
r_o^{(k)}(t) &\coloneqq [r_{o,1}^\top(t),\dots,r_{o,n_k}^\top(t)]^\top \in \mathbb{R}^{dn_k}, \nonumber\\
\rho_o^{(k)}(t) &\coloneqq (E_k^\top\otimes I_d)r_o^{(k)}(t) \in \mathbb{R}^{dm_k}. 
\end{align}
The vector $\rho_o^{(k)}(t)$ compactly represents the collection of edge-wise drift-center differences $r_{o,i}(t)-r_{o,j}(t)$ for all $(i,j)\in\mathcal{E}_k$ up to the arbitrary edge orientation.

The stabilization objective is formulated as suppressing $\rho_o^{(k)}(t)$ so that disagreement among $\{r_{o,i}(t)\}_{i\in\mathcal{V}_k}$ vanishes on the connected cluster.
When the actuator is not saturated and the feedback gains are appropriately tuned for the anticipated disturbance environment, the closed-loop effect of the distributed controller can be interpreted as a consensus-type contraction of the drift-center differences on $\mathcal{G}_k$.
Namely, the controller reduces disagreement in $\{r_{o,i}(t)\}$ and drives the cluster toward a configuration where all satellites share approximately the same drift center.
In this regime, once the drift-center disagreement is sufficiently suppressed, the orbital dynamics no longer induce unbounded separation among satellites, and the formation remains bounded without requiring continuous large corrective actions.

\subsection{Gaussian-distribution-based probabilistic safety guarantees}
\label{subsec:prelim_gaussian_prob}

This subsection summarizes only the probabilistic facts that are directly used to derive the deterministic sufficient condition for the stage-wise chance constraint in later sections. Throughout, $\|\cdot\|_2$ denotes the Euclidean norm and $\lambda_{\max}(\cdot)$ denotes the maximum eigenvalue of a symmetric matrix.

Let $\bm{x}\sim\mathcal{N}(\bm{\mu},\bm{\Sigma})$ and $\bm{y}\sim\mathcal{N}(\bm{\mu}_y,\bm{\Sigma}_y)$ be independent Gaussian random vectors. For an affine combination $\bm{z}=A\bm{x}+B\bm{y}$ with deterministic matrices $(A,B)$, $\bm{z}$ is also Gaussian and its moments are given by
\begin{equation}
\label{eq:prelim_gaussian_affine_compact}
\mathbb{E}[\bm{z}] = A\bm{\mu}+B\bm{\mu}_y,\quad
\mathrm{Cov}(\bm{z}) = A\bm{\Sigma}A^\top + B\bm{\Sigma}_yB^\top .
\end{equation}

Next, consider the Euclidean-norm chance constraint $\mathbb{P}(\|\bm{x}\|_2\le r_c)\ge 1-\beta$ with $\bm{x}\sim\mathcal{N}(\bm{\mu},\bm{\Sigma})$ and $\bm{\Sigma}\succ 0$. The centered quadratic form $(\bm{x}-\bm{\mu})^\top\bm{\Sigma}^{-1}(\bm{x}-\bm{\mu})$ follows a chi-square distribution with $d$ degrees of freedom, and let $\chi^2_{d,\,1-\beta}$ denote its $(1-\beta)$ quantile. Using the Rayleigh-quotient bound 
\begin{align}
    \bm{z}^\top\bm{\Sigma}\bm{z}\le \lambda_{\max}(\bm{\Sigma})\|\bm{z}\|_2^2
\end{align} 
a $(1-\beta)$ confidence ellipsoid is contained in the Euclidean ball centered at $\bm{\mu}$ with radius
\begin{equation}
\label{eq:prelim_confidence_ball_radius}
r(\bm{\Sigma},\beta)\coloneqq
\sqrt{\chi^2_{d,\,1-\beta}\;\lambda_{\max}(\bm{\Sigma})}.
\end{equation}
Therefore, a sufficient deterministic condition for the chance constraint is
\begin{equation}
\label{eq:prelim_chance_sufficient_compact}
\|\bm{\mu}\|_2 + r(\bm{\Sigma},\beta) \le r_c,
\end{equation}
which implies $\mathbb{P}(\|\bm{x}\|_2>r_c)\le \beta$.
This confidence-ball reduction is the only probabilistic ingredient needed in the subsequent stage-wise safety derivation.

\section{Problem Formulation}
\label{sec:problem_formulation}

This section formulates the sequential initialization of a satellite swarm as a stage-wise design problem on an expanding interaction graph.
The objective is to determine a deployment policy and interval so that newly formed inter-satellite links remain controllable/safe at each deployment stage with a prescribed risk level.

\subsection{Sequential deployment and objective}
\label{subsec:sequential_objective}

We consider a mission scenario in which a satellite swarm is constructed sequentially by releasing satellites from a carrier.
Let $N$ be the total number of satellite clusters to be deployed.
Deployment events occur at discrete time instants
\begin{align}
t_k \coloneqq k\Delta T,\qquad k=0,1,\dots,N-1.
\end{align}
where $\Delta T>0$ is the deployment interval.

The swarm topology is modeled as a sequence of expanding graphs $\{\mathcal{G}_k\}_{k=1}^{N}$.
At stage $k$, the currently deployed and controllably connected cluster is represented by $\mathcal{G}_k=(\mathcal{V}_k,\mathcal{E}_k)$.
At the next deployment time $t_{k+1}$, a set of new satellites $\mathcal{W}_{k+1}$ is added, and a set of new edges $\mathcal{F}_{k}$ is established to connect the new satellites to the existing cluster.
Accordingly, the graph expands as Eq.~\eqref{eq:graph_expansion_pf_revised} and Figure~\ref{fig:sequentialdeployment}
\begin{align}
\label{eq:graph_expansion_pf_revised}
\mathcal{V}_{k+1} = \mathcal{V}_k \cup \mathcal{W}_{k+1},\qquad
\mathcal{E}_{k+1} = \mathcal{E}_k \cup \mathcal{F}_{k+1}.
\end{align}

\begin{figure}
    \centering
    \includegraphics[width=\linewidth]{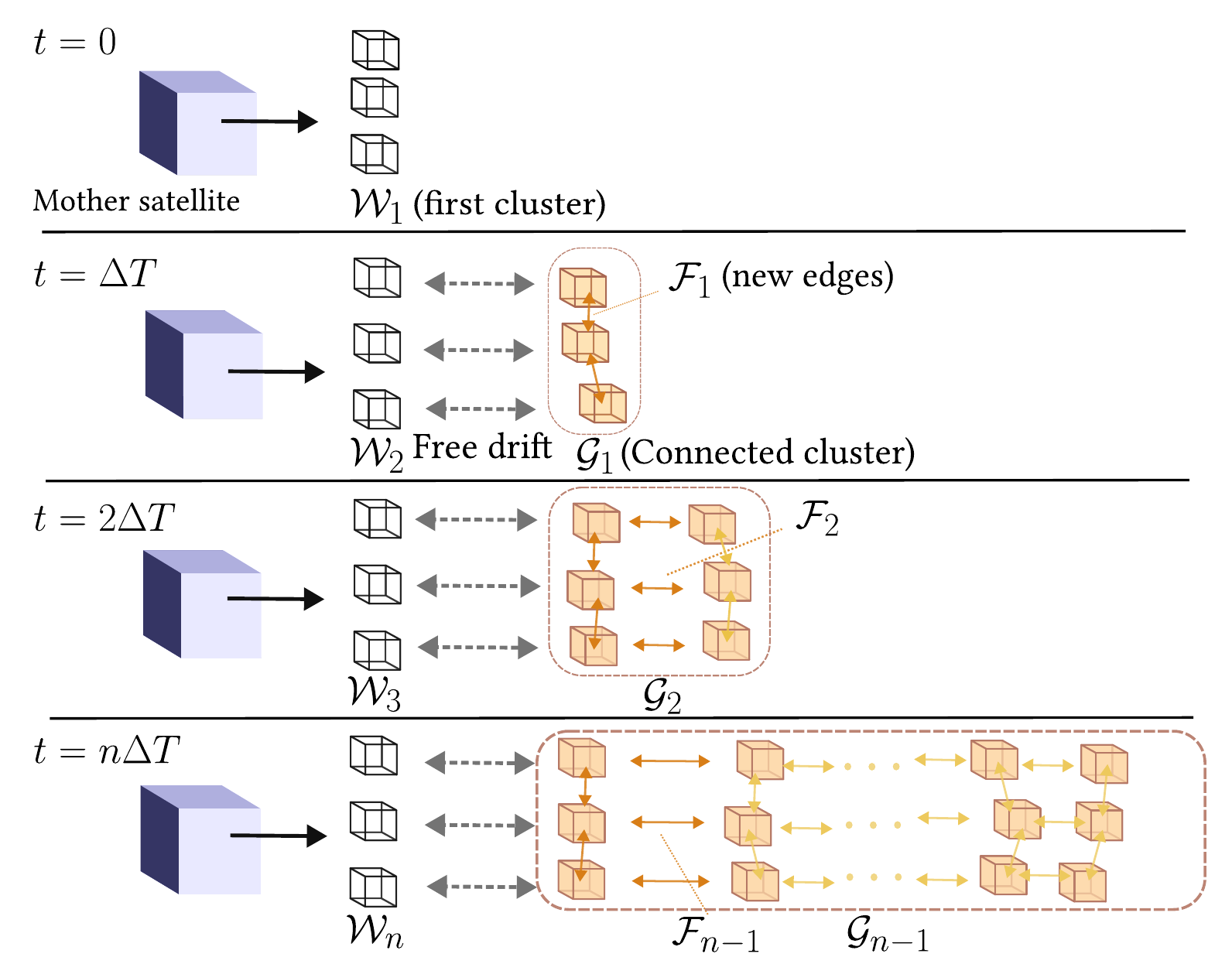}
    \caption{A graphical image of sequential deployment}
    \label{fig:sequentialdeployment}
\end{figure}

During each interval $[t_k,t_{k+1})$, satellites already in the cluster follow the closed-loop distributed controller described earlier, whereas newly released satellites may undergo a stage-dependent free-drift before their links are activated.
As a compact representation of inter-satellite mismatch on the currently active links, the edge state is defined as the edge-wise drift-center difference on $\mathcal{G}_k$, as in Eq.~\eqref{eq:stack_ro}.

The design objective is to choose the deployment interval $\Delta T$ and a stage-wise deployment/attachment policy (equivalently, the sequence of attachment sets $\{\mathcal{F}_{k+1}\}_{k=0}^{N-1}$ and the associated anchor assignments) so that the stage-wise safety requirement stated in the next subsection is satisfied for all expansion steps $k\to k+1$.

\subsection{Deployment-induced drift-center mismatch model}
\label{subsec:deployment_error_model}

This subsection models how the drift-center mismatch evolves during the interval between two deployment instants.
In the sequential initialization considered here, satellites already integrated into the controllably connected cluster reduce their drift-center disagreement through distributed feedback, whereas newly released satellites may undergo a free-drift period before their links are activated.
The following model captures these two effects at the drift-center level.

For the existing cluster represented by $\mathcal{G}_k$, the stacked edge-wise drift-center difference is assumed to contract over one deployment interval as
\begin{align}
\label{eq:pf_contraction_model}
\bm\rho^{(k)}(t_{k+1}^-)
=
\Phi_k\,\bm\rho^{(k)}(t_k^+),\notag\\
\Phi_k\coloneqq \exp(-\Delta T\,(L_{e,k}\otimes A)).
\end{align}
where $L_{e,k}$ is the edge Laplacian of $\mathcal{G}_k$ introduced in the preliminaries and $A$ represents the closed-loop contraction dynamics associated with the consensus regulation of the drift-center coordinates.

Next, consider a newly created edge between a leading satellite $i$ and a trailing satellite $j$.
Because the activation times of effective control are not synchronized, the two satellites experience different free-drift durations before the link becomes active.
Let $\Psi(t)$ denote the state-transition matrix of the drift-center dynamics under free drift over duration $t$.
Then the drift-center mismatch injected into the new edge at the activation instant is modeled as the difference of the two free-drift trajectories:
\begin{align}
\label{eq:pf_w_general}
\bm{w}^{(k+1)}_{i,j}
=
\Psi(\tau_i)\,\bm{r}_{o,i}(t_i^+)
-
\Psi(\tau_j)\,\bm{r}_{o,j}(t_j^+),
\end{align}
where $\tau_i$ and $\tau_j$ are the free-drift durations of satellites $i$ and $j$, respectively.
For the drift-center coordinate used in this work, $\Psi(t)$ is given by
\begin{align}
\label{eq:pf_Psi_def}
\Psi(t)
\coloneqq
\begin{bmatrix}
1 & 0\\[2pt]
-\dfrac{\epsilon_2}{2}t & 1
\end{bmatrix},
\end{align}
In the representative timing where the leading satellite experiences one additional deployment interval of free drift, the injected mismatch can be written as
\begin{align}
\label{eq:pf_w_asym}
\bm{w}^{(k+1)}_{i,j}
=
\Psi(2\Delta T)\,\bm{r}_{o,i}(t_k^+)
-
\Psi(\Delta T)\,\bm{r}_{o,j}(t_{k+1}^+).
\end{align}

When the area-to-mass ratio is large, atmospheric drag may add an extra drift component during free drift.
To incorporate this effect, the drift-center motion is corrected by a drag-induced component.
Let $\bm{r}_o(t)$ denote the drift-center motion defined under the averaged $J_2$ model, and let $\bm{r}'_o(t)$ denote the corrected drift-center motion including atmospheric drag.
Assuming linear superposition, the corrected drift-center is written as
\begin{align}
\label{eq:pf_ro_superposition}
\bm{r}'_o(t)=\bm{r}_o(t)+\bm{r}_{\mathrm{air},o}(t).
\end{align}
Mapping the constant term and the drift term proportional to $t$ to the relative-element pair $(C_1,C_4)$ yields
\begin{align}
\label{eq:pf_ro_prime}
\bm{r}'_o(t)
&=
\begin{bmatrix}
2C'_1 \\[2pt] C'_4 -\epsilon_2\, C'_1\, t
\end{bmatrix}
\notag\\
&=
\begin{bmatrix}
2\!\left(C_1 + \dfrac{1}{2}C_{1,\mathrm{air}}\right)\\[6pt]
C_4 + \dfrac{\epsilon_2}{2}\,C_{4,\mathrm{air}}
- \epsilon_2\!\left(C_1+\dfrac{1}{2}C_{1,\mathrm{air}}\right)t
\end{bmatrix}.
\end{align}
Accordingly, the injected mismatch model Eq.~\eqref{eq:pf_w_asym} is extended by replacing $\bm{r}_{o,i}$ and $\bm{r}_{o,j}$ with $\bm{r}'_{o,i}$ and $\bm{r}'_{o,j}$.

Here, the atmospheric-drag-induced increments $C_{1,\mathrm{air}}$ and $C_{4,\mathrm{air}}$ are defined as
\begin{align}
\label{eq:C_air_defs_theta_xi_edge}
  C_{1,\mathrm{air}}
  &\coloneqq \sum_{m=1}^{\infty}
     \frac{\hat F_m\,A/m}{c_-\,\nu_m}\,\sin{\psi_m}, \\
  C_{4,\mathrm{air}}
  &\coloneqq \sum_{m=1}^{\infty}
     \frac{\hat F_m\,A/m}{c_-\,\nu_m^{2}}\,\cos{\psi_m}.
\end{align}
where $\nu_m$ is a variable depending on the mean value of the angular rate. The detail derivation of the atmospheric-drag-induced increments is explained in Appendix~A.
Accordingly, the deployment-induced error is described using the extended drift center as
\begin{align}
\label{eq:w_asym_with_drag_theta_xi_edge}
w^{k+1}
= \Psi(2\Delta T)\,\bm{r}_{o,i}'(0) - \Psi(\Delta T)\,\bm{r}_{o,j}'(0).
\end{align}

\subsection{Chance-constrained safety requirement}
\label{subsec:safety_requirement}

The closed-loop model and distributed feedback assumptions are intended to hold only when each newly activated link starts within a domain where control and inter-satellite interaction are effective.
Accordingly, a controllable/safe region for each edge state is specified by the Euclidean ball
\begin{align}
\label{eq:safe_region_pf_revised}
\mathcal{S}_{\mathrm{cont}}
\coloneqq
\left\{\bm{\rho}_{o,e}\in\mathbb{R}^{d}\ \middle|\ \|\bm{\rho}_{e}\|_2 \le r_c \right\}.
\end{align}
where $r_c>0$ is an effective control/communication radius.

The most critical moments in sequential deployment occur immediately after new connections are formed, i.e., at $t_{k+1}^+$.
At these instants, newly generated edges inherit offsets accumulated during free drift.
To account for stochastic deployment errors and disturbances, stage-wise safety is imposed in a chance-constrained manner.
For each expansion step $k\to k+1$, every newly formed edge is required to satisfy the controllability condition with probability at least $1-\beta$:
\begin{align}
\label{eq:edgewise_stage_chance_pf_revised}
\max_{e\in\mathcal{F}_{k+1}}
\mathbb{P}\!\left(\|\bm{\rho}_{e}^{(k+1)}(t_{k+1}^+)\|_2 > r_c\right)
\le \beta,\notag\\\qquad k=0,1,\dots,N-1.
\end{align}
where $\mathcal{F}_{k+1}\subseteq\mathcal{E}_{k+1}$ denotes the set of newly created edges at stage $k+1$ and $\beta\in(0,1)$ is the allowable risk level per edge per stage.
A deterministic sufficient condition for Eq.~\eqref{eq:edgewise_stage_chance_pf_revised} will be derived later using the Gaussian-probabilistic tools summarized in the preliminaries.
\section{Main Results: Probabilistic Release Sequence Design}
\label{sec:main_results}

This section presents two results.
First, a stage-to-stage linear recursion is derived for the stacked edge-wise drift-center mismatch when the interaction graph expands through sequential deployment.
Second, a deterministic sufficient condition is obtained to enforce the stage-wise chance-constrained safety requirement introduced in Section~\ref{subsec:safety_requirement}.

\subsection{Recursive error propagation formulation for an expanding graph}
\label{subsec:expanding_graph_recursion}

This subsection derives a compact update rule for the stacked edge state across an expansion step.
Under the contraction model for the existing cluster and the injected mismatch model for newly created edges specified in Section~\ref{subsec:deployment_error_model}, the post-expansion edge state is expressed as an affine function of the pre-expansion state and the injected mismatch.

\begin{lemma}[Stage-to-stage update of the stacked edge state]
\label{lem:expanding_graph_update}
Consider an expansion step $\mathcal{G}_k\to\mathcal{G}_{k+1}$ with newly created edges $\mathcal{F}_{k+1}$.
Let the existing-cluster contraction over one deployment interval be as in \eqref{eq:pf_contraction_model}, and let $\bm w^{(k+1)}\in\mathbb{R}^{d|\mathcal{F}_{k+1}|}$ denote the stacked drift-center mismatch injected into the newly created edges at $t_{k+1}^+$ as described in Section~\ref{subsec:deployment_error_model}.
Then the stacked edge state immediately after the graph expansion satisfies
\begin{align}
\label{eq:rho_update_main}
\bm\rho^{(k+1)}(t_{k+1}^+)
=
\mathcal{A}_k\,\bm\rho^{(k)}(t_k^+)
+
\mathcal{B}_k\,\bm w^{(k+1)}.
\end{align}
Here, letting $\mathcal{P}_{k+1}$ be a permutation matrix that aligns the edge ordering of $\mathcal{G}_{k+1}$, define
\begin{align}
\label{eq:Ak_def_main}
\mathcal{A}_k
&\coloneqq
\mathcal{P}_{k+1}
\begin{bmatrix}
\Phi_k\\
\mathcal{R}_{k,\mathcal{F}_{k+1}}(I_{dm_k}-\Phi_k)
\end{bmatrix},\\
\label{eq:Bk_def_main}
\mathcal{B}_k
&\coloneqq
\mathcal{P}_{k+1}
\begin{bmatrix}
0\\
I_{d|\mathcal{F}_{k+1}|}
\end{bmatrix},
\end{align}
and the projection matrix onto the newly created edges is
\begin{align}
\label{eq:R_def_main}
\mathcal{R}_{k,\mathcal{F}_{k+1}}
\coloneqq
\bigl(V_{k,\mathcal{F}_{k+1}}^\top E_k\, L_{e,k}^\dagger \otimes I_d\bigr).
\end{align}
where $V_{k,\mathcal{F}_{k+1}}\in\{0,1\}^{n_k\times|\mathcal{F}_{k+1}|}$ specifies to which node (anchor) in the existing cluster each new edge is attached and $\Phi_k$ is given in Eq.~\eqref{eq:pf_contraction_model}.
\end{lemma}

\begin{proof}\noindent
The proof is provided in Appendix~B.
\end{proof}

\subsection{Deterministic sufficient condition for stage-wise chance-constrained safety}
\label{subsec:prob_safety_condition}

This subsection converts the probabilistic edge-safety requirement in Eq.~\eqref{eq:edgewise_stage_chance_pf_revised} into a tractable deterministic inequality.
The key step is to combine the affine recursion in Lemma~\ref{lem:expanding_graph_update} with the Gaussian tools summarized in Section~\ref{subsec:prelim_gaussian_prob}.

For each $e\in\mathcal{E}_{k+1}$, introduce a selection matrix $J_{k+1,e}\in\{0,1\}^{d\times dm_{k+1}}$ that extracts the $d$-dimensional block associated with edge $e$ from a stacked vector in $\mathbb{R}^{dm_{k+1}}$.
Then
\begin{align}
\label{eq:edge_extract_main}
\bm\rho^{(k+1)}_e(t_{k+1}^-)
=
J_{k+1,e}\,\bm\rho^{(k+1)}(t_{k+1}^+).
\end{align}

Assume that $\bm\rho^{(k)}(t_k^+)$ and $\bm w^{(k+1)}$ are independent Gaussian random vectors.
Then, by Lemma~\ref{lem:expanding_graph_update}, the updated state $\bm\rho^{(k+1)}(t_{k+1}^+)$ is also Gaussian.
The mean and covariance of $\bm\rho^{(k+1)}(t_{k+1}^+)$ follow from the affine-Gaussian rule in Section~\ref{subsec:prelim_gaussian_prob}. In particular, letting $(\bm{\mu}_{\rho,k},\bm{\Sigma}_{\rho,k})$ denote the mean and covariance of $\bm\rho^{(k)}(t_k^+)$, and letting $(\bm{\mu}_{w,k+1},\bm{\Sigma}_{w,k+1})$ denote those of $\bm w^{(k+1)}$, the moments of $\bm\rho^{(k+1)}(t_{k+1}^+)$ are updated as
\begin{align}
\label{eq:mean_update_edgewise_theta_xi}
\bm{\mu}_{k+1} &= \mathcal{A}_k \bm{\mu}_{\rho,k} + \mathcal{B}_k \bm{\mu}_{w,k+1},\\
\label{eq:cov_update_edgewise_theta_xi}
\bm{\Sigma}_{k+1} &= \mathcal{A}_k \bm{\Sigma}_{\rho,k} \mathcal{A}_k^\top
+ \mathcal{B}_k \bm{\Sigma}_{w,k+1}\mathcal{B}_k^\top .
\end{align}
Let $(\bm\mu_{k+1},\bm\Sigma_{k+1})$ denote the mean and covariance of $\bm\rho^{(k+1)}(t_{k+1}^+)$.
Define the edge-wise moments by
\begin{align}
\label{eq:edge_moments_main}
\bm\mu_{k+1,e}\coloneqq J_{k+1,e}\bm\mu_{k+1},\quad
\bm\Sigma_{k+1,e}\coloneqq J_{k+1,e}\bm\Sigma_{k+1}J_{k+1,e}^\top .
\end{align}

This theorem provides a sufficient condition that enforces Eq.~\eqref{eq:edgewise_stage_chance_pf_revised} by bounding the edge-state distribution using the confidence-ball argument in Section~\ref{subsec:prelim_gaussian_prob}.
The condition depends only on the edge-wise mean and covariance at the activation instant.

\begin{theorem}[Stage-wise chance-constrained safety condition]
\label{thm:stagewise_safety_main}
Consider the expansion step $k\to k+1$ and the set of newly created edges $\mathcal{F}_{k+1}$.
If, for all $e\in\mathcal{F}_{k+1}$,
\begin{align}
\label{eq:sufficient_condition_main}
\|\bm\mu_{k+1,e}\|_2
+\sqrt{\chi^2_{d,\,1-\beta}\;\lambda_{\max}(\bm\Sigma_{k+1,e})}
\le r_c,
\end{align}
then the requirement Eq.~\eqref{eq:edgewise_stage_chance_pf_revised} is satisfied, namely
\begin{align}
\label{eq:prob_bound_main}
\forall e\in\mathcal{F}_{k+1},\qquad
\mathbb{P}\!\left(\|\bm\rho^{(k+1)}_e(t_{k+1}^+)\|_2> r_c\right)\le \beta.
\end{align}
\end{theorem}

\begin{proof}\noindent
The proof is provided in Appendix~C.
\end{proof}

This corollary summarizes a closed-form representation that is convenient when the deployment policy is fixed.
It enables batch evaluation of edge safety at a specified stage by propagating the influence of the injected mismatches.

\begin{corollary}[Closed-form evaluation for a fixed deployment policy]
\label{cor:closed_form_main}
Let the deployment policy (and thus $\{\mathcal{A}_k,\mathcal{B}_k\}_{k=0}^{M-1}$) be fixed.
Then the stacked edge state at stage $M$ admits the representation
\begin{align}
\label{eq:closed_form_rho_main}
\bm\rho^{(M)}(t_M^+)
&=
\left(\prod_{j=0}^{M-1}\mathcal{A}_j\right)\bm\rho^{(0)}(t_0^+)
+
\sum_{i=1}^{M}\Gamma_{M,i}\bm w^{(i)},\\
\label{eq:Gamma_def_main}
\Gamma_{M,i}
&\coloneqq
\left(\prod_{l=i}^{M-1}\mathcal{A}_l\right)\mathcal{B}_{i-1}.
\end{align}
\end{corollary}

\section{Numerical Simulation}
\subsection{Simulation Setup}
This section verifies the proposed probabilistic safety design for a sequential release sequence.
A numerical study is conducted by sweeping the release interval $\Delta T$.

The objective is to evaluate how the allowable release-error covariance changes with $\Delta T$.
A safety radius $r_c$ and a failure probability $\beta$ are prescribed.
The corresponding success probability is $1-\beta$.

Two release-condition settings are compared.
In setting (i), the nominal free drift $\|\bm{\mu}_{w,e}\|$ (in Eq.~\eqref{eq:mean_update_edgewise_theta_xi}) is kept constant by fixing the initial relative velocities.
In setting (ii), the initial release velocity is adjusted so that $\|\bm{\mu}_{w,e}\|$ stays constant for any $\Delta T$.

The simulation considers a sequence in the LVLH frame.
A row of three satellites aligned along the $y$-axis is used as a basic unit (width $3$).
These rows are released sequentially in the $+\!x$ direction with a constant interval $\Delta T$.

To represent the evolving interaction structure, the Laplacian at the start of control for the $k$th row is denoted by $L_k$.
It is given by a block tridiagonal form:
\begin{align*}
    L_k =
    \begin{bmatrix}
        L_a & L_b & \bm{0} \\
        L_b & L_a & L_b & & \dots & \bm{0} \\
        \bm{0} & L_b & L_a & L_b \\
        & & L_b & L_a \\
        & \vdots & & & \ddots & L_b \\
        & \bm{0} & & & L_b & L_a
    \end{bmatrix}.
\end{align*}
Here, $L_a$ represents intra-row (horizontal) coupling.
The matrix $L_b$ represents inter-row (vertical) coupling:
\begin{align*}
    L_a &=
    \begin{bmatrix}
        0 & 1 & 0 \\
        1 & 0 & 1 \\
        0 & 1 & 0
    \end{bmatrix},
    \qquad
    L_b =
    \begin{bmatrix}
        1 & 0 & 0 \\
        0 & 1 & 0 \\
        0 & 0 & 1
    \end{bmatrix}.
\end{align*}

The distance between the centers of mass of adjacent satellites within a row is set to $0.25~\mathrm{m}$ along the $y$-axis.
The release speed is assumed to be identical for all satellites.

The closed-loop matrix based on consensus control is defined as
\begin{align*}
    A = \frac{k_A}{k_0}
    \begin{bmatrix}
        1 & 0 \\
        0 & 1
    \end{bmatrix}.
\end{align*}
where, a control gain is defined as $k_A$ and a orbital motion constant is defined as $k_0\coloneqq \frac{2c_+}{\omega_{xy}c_-}$. Detail formulation of this closed-loop matrix is explained in \cite{takahashi2025scalable}.
The effective separation distance is set to $r_c = 1.0~\mathrm{m}$.
The failure probability is set to $\beta = 0.01$ to guarantee a $99\%$ success probability.

The atmospheric density is assumed to be constant.
A low Earth orbit is used as the reference orbit to emulate an ISS deployment scenario.
Tip-off rotation at release is also considered.
The angular rate is estimated by assuming that the line of action of $\Delta V$ is offset by $0.01~\mathrm{m}$ from the satellite center of mass.

All parameters are summarized in Table~\ref{tab:numcase_params}.

\begin{table}[h]
  \centering
  \caption{Numerical case study parameters}
  \label{tab:numcase_params}
  \begin{tabular}{lll}
    \hline
    \textbf{Description} & \textbf{Symbol} & \textbf{Value} \\
    \hline
    Earth GM & $\mu$ & $3.99\times10^{14}\ \mathrm{m^3/s^2}$ \\
    Earth radius & $R_e$ & $6.37\times10^{3}\ \mathrm{km}$ \\
    Altitude & $h$ & $4.00\times10^{2}\ \mathrm{km}$ \\
    Radius & $R_t$ & $R_e+h$ \\
    Inclination & $i_0$ & $51.7^\circ$ \\
    Atmospheric density & $\rho$ & $1.18\times10^{-12}\ \mathrm{kg/m^3}$ \\
    Drag coefficient & $C_d$ & $2.00$ \\
    Area-to-mass & $A/m$ & $1.00\times10^{-2}\ \mathrm{m^2/kg}$ \\
    Mass & $m$ & $1.0\ \mathrm{kg}$ \\
    Size & $\ell$ & $0.10\ \mathrm{m}$ \\
    Inertia (cube) & $I$ & $m\ell^2/6$ \\
    \hline
  \end{tabular}
\end{table}

In this case study, the allowable release-error dispersion is evaluated as a function of $\Delta T$ under the safety constraints.
Two release-condition settings are considered.
In setting (i), the initial relative velocities are fixed as $\dot{\bm{x}}=\dot{\bm{y}}=0.001~\mathrm{m/s}$.
In setting (ii), the initial release velocity is tuned according to the interval $\Delta T$. Specifically, the velocity is adjusted to keep the varying drift component ($C_1 \Delta T$ inEq.~\eqref{eq:w_asym_with_drag_theta_xi_edge}) constant, thereby ensuring that the drift distance induced by the release velocity remains invariant across all intervals. This setup is intended to isolate and verify the influence of release-error dispersion.

The release-error dispersion is defined as a relative deviation from the nominal values of the drift-center components $[2C'_1,\;C'_4]$.
The allowable variance factor is organized using this deviation ratio.
This definition accounts for the fact that the nominal values depend on the drift condition.
Therefore, the error tolerance is reported as the percentage dispersion relative to the nominal values.

\subsection{Deployment Error Analysis}
First, consider case (i), where the initial relative velocity is kept constant.
Figure~\ref{fig:case1_safetyresult} shows the requirements on the release-error magnitude and the release-to-activation interval that satisfy the safety constraint.
The requirements become stricter as the number of clusters increases from $50$ to $100$ and $300$.

To illustrate the safety budget, a representative condition is selected.
The condition is $N=100$ clusters ($3N_\text{cluster}=300$ satellites) with a relative deviation of $5\%$.
Figure~\ref{fig:case1_safetybudget} plots two contributions in Eq. \eqref{eq:cov_update_edgewise_theta_xi} as functions of $\Delta T$.
The first is the edge component of the nominal drift, $\|A_k\bm{\mu}_\rho\|$.
The second is the injected-error mean term, $\|B_k\bm{\mu}_w\|$.
Figure~\ref{fig:case1_safetybudget} indicates that $\bm{\mu}_w$ increases in proportion to the time interval.
This term dominates the safety budget for large $\Delta T$.
This trend is consistent with the $C'_1\Delta T$ term in Eq.~\eqref{eq:w_asym_with_drag_theta_xi_edge}.

Next, case (ii) is examined.
Figure~\ref{fig:case2_safetyresult} and Figure~\ref{fig:case2_safetybudget} show a clear difference from Figure~\ref{fig:case1_safetyresult}.
An optimal interval $\Delta T$ exists that maximizes the allowable variance factor.
This behavior can be explained by two competing effects.
If $\Delta T$ is too small, the effect of initial release errors in the already-released satellites is not sufficiently attenuated.
Conversely, if $\Delta T$ is too large, the nominal drift consumes a larger portion of the safety budget.
This is driven by two factors: the displacement induced by the consensus control becomes larger, and the injected drift magnitude $\|\mathcal{B}_{k}\bm{\mu}_{w}\|$ increases because the reduced angular rate (due to lower release velocity) amplifies the aerodynamic drag coefficient $C_{1,air}$.
As a result, the allowable error decreases as interval $\Delta T$ increases.
Thus, considering the influence of atmospheric drag associated with the initial release angular rate, it is demonstrated that optimal release intervals and velocities exist to maximize the allowable release error dispersion.

\begin{figure}[ht]
    \centering
    \includegraphics[width=1\linewidth]{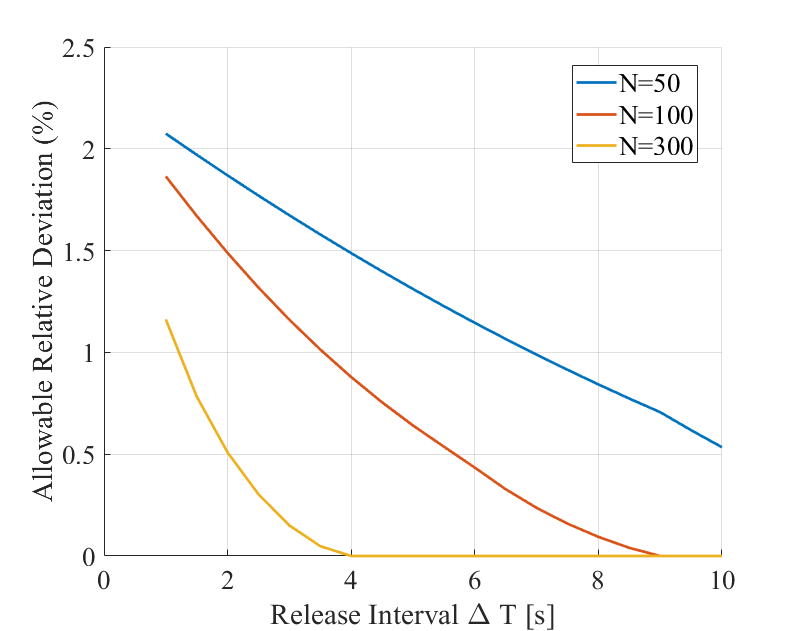}
    \caption{Safe initialization conditions with fixed release velocities in case (i)}
    \label{fig:case1_safetyresult}
\end{figure}

\begin{figure}[ht]
    \centering
    \includegraphics[width=1\linewidth]{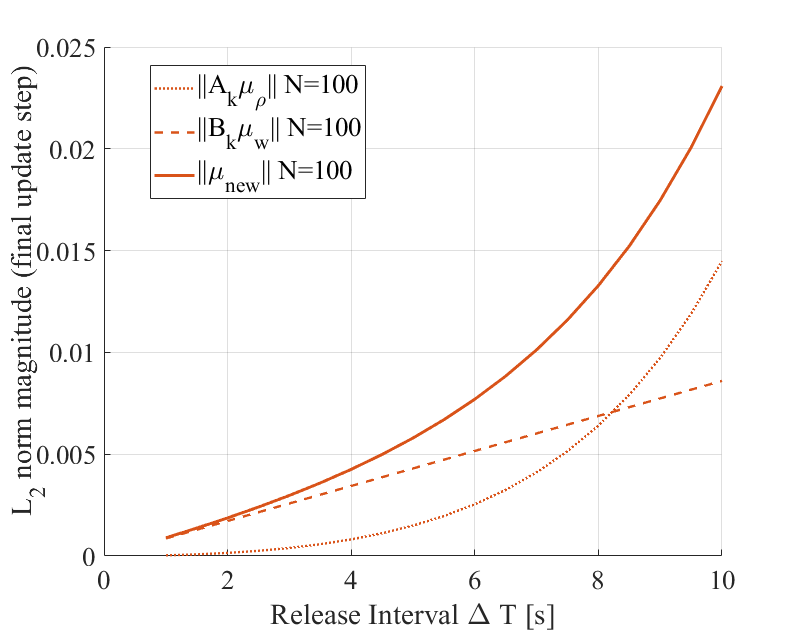}
    \caption{Ratio of control-induced drift to uncontrolled drift in case (i)}
    \label{fig:case1_safetybudget}
\end{figure}

\begin{figure}[ht]
    \centering
    \includegraphics[width=1\linewidth]{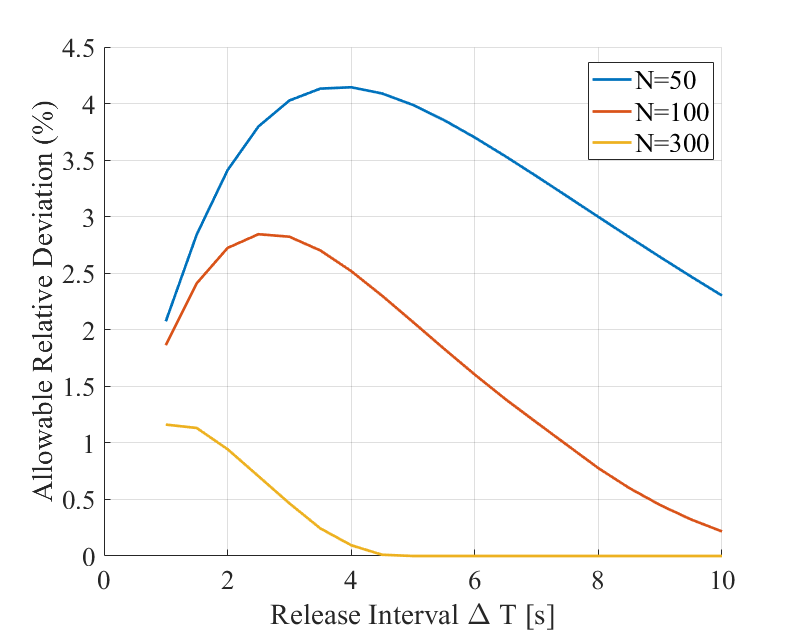}
    \caption{Safe initialization conditions with tuned release velocities in case (ii)}
    \label{fig:case2_safetyresult}
\end{figure}

\begin{figure}[ht]
    \centering
    \includegraphics[width=1\linewidth]{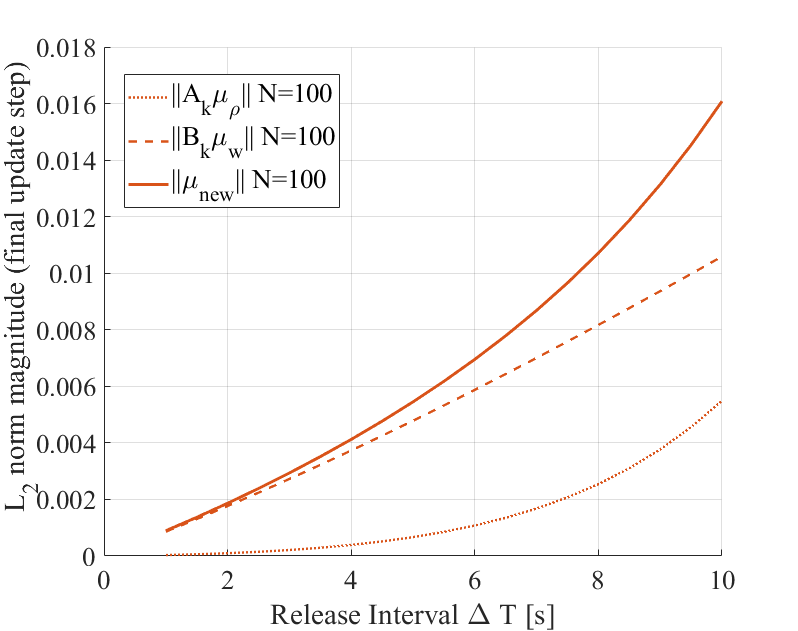}
    \caption{Ratio of control-induced drift to uncontrolled drift in case (ii)}
    \label{fig:case2_safetybudget}
\end{figure}

Finally, a Monte Carlo validation is performed for case (ii).
The condition is $N=100$ clusters with $\Delta T = 4~\mathrm{s}$ and a variance factor of $2.5\%$.
A total of $1000$ trials are simulated.
The result is shown in Figure~\ref{fig:montecalro}.

In Figure~\ref{fig:montecalro}, the largest drift-center distance of the new edges in the worst 100 cases is plotted at each time.
In this validation, the number of trials was limited to 1000 due to computational resource constraints.
While this sample size is not exhaustive enough to strictly verify the theoretical tail distribution, the fact that zero failures were observed across all 1000 trials empirically supports the validity of the proposed method.
This result confirms that the system achieves a level of stability consistent with the design requirement ($1-\beta = 99\%$) with high confidence.
\begin{figure}[h]
    \centering
    \includegraphics[width=1\linewidth]{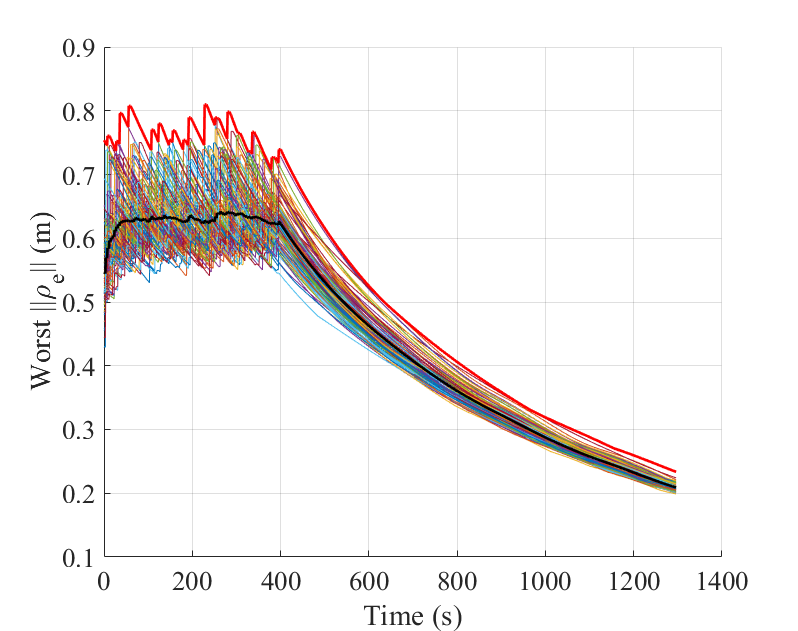}
    \caption{Monte-Calro evaluation of safe initialization condition $r_c = 1.0 ~\mathrm{m}$ (Worst 100 cases: red line = the worst case, black line = the mean value of the 100 cases)}
    \label{fig:montecalro}
\end{figure}
\section{Conclusions and Discussions}
This study proposes probabilistic constraint conditions for release design in swarm initialization.
The target scenario includes a limited control and communication distance after release.
A scalable framework is established to achieve swarm initialization with a prescribed probability.

The proposed method covers both the release phase and the subsequent consensus control phase.
This is a key difference from conventional designs that assume only an uncontrolled drift interval.
The results indicate that an optimal release interval exists.
This interval maximizes the allowable tolerance to hardware-induced release errors that cannot be eliminated.

A remaining issue is that the condition for the interval that maximizes the allowable release error has not been derived.
The optimal interval is expected to depend on the graph structure and the control gain $k_A$.
A detailed analytical derivation will be addressed in future work.

Several extensions are also possible.
One direction is validation under a broader variety of graph structures.
Another direction is to examine whether a minimum number of satellites (clusters) is required for reliable verification of scalability.
By addressing these issues, a more general and widely applicable initialization design can be developed.
\section*{Appendices}
\subsection*{Appendix A. Aerodynamic-drag-induced increments in the drift-center coefficients}
This appendix summarizes how the aerodynamic-drag contribution during the uncontrolled interval is incorporated into the drift-center representation through the increments $C_{1,\mathrm{air}}$ and $C_{4,\mathrm{air}}$.
The starting point is an along-track differential drag forcing term whose magnitude depends on the projected area, which can vary in time when the spacecraft is released with a residual angular rate.

Focusing on the short time window of initial deployment enables an analytical estimate of the aerodynamic-drag contribution. The relative motion is driven by the differential area-to-mass ratio; to isolate this effect, the analysis considers the dynamics with respect to a virtual satellite of zero projected area and then forms relative quantities. Within the LVLH frame, the along-track velocity is assumed dominant and aerodynamic lift is neglected, so the along-track drag forcing is modeled as a scalar input $F_{\mathrm{air}}(t)$ with $u_y(t)=-c_-F_{\mathrm{air}}(t)$. For a cube of side length $a$ spinning about a face-normal principal axis with attitude phase $\phi$ and spin rate $\nu$ at $t=0$, the time-varying along-track projection yields
\begin{align}
F_{\mathrm{air}}(t)=\frac{k_{\mathrm{air}}}{2}\,\frac{a^2}{m}\Bigl(\lvert \cos(\nu t+\phi)\rvert+\lvert \sin(\nu t+\phi)\rvert\Bigr),
\label{eq:Fair_abs_app}
\end{align}
and the absolute values are expanded into a Fourier series and treated as a sum of linear trigonometric forcings. Using the standard expansions of $|\sin x|$ and $|\cos x|$, the forcing comprises a DC component and even harmonics only, and it can be written as
\begin{align}
F_{\mathrm{air}}(t)&=F_{\mathrm{ref}}\left[\frac{\pi}{4}-\frac{\pi}{2}\sum_{m=1}^{\infty}\frac{\cos\!\bigl(4m(\nu t+\phi)\bigr)}{16m^2-1}\right],\\
F_{\mathrm{ref}}&\coloneqq\frac{8}{\pi^2}k_{\mathrm{air}}\frac{a^2}{m},
\label{eq:Fair_fourier_app}
\end{align}
so that
\begin{align}
F_{\mathrm{DC}}&=\frac{\pi}{4}F_{\mathrm{ref}},\quad
&&\hat F_m=\frac{\pi}{2}\frac{F_{\mathrm{ref}}}{16m^2-1},\notag\\
\nu_m&\coloneqq 4m\nu,\quad &&\psi_m\coloneqq 4m\phi.
\label{eq:Fair_coeffs_app}
\end{align}
The averaged-$J_2$ HCW model is linear, so the particular solution is built by superposition of the responses to the DC term and to each even harmonic. Let $\mathcal{L}\{f(t)\}=F(s)$ denote the Laplace transform and consider the in-plane averaged-$J_2$ HCW equations for the particular solution under zero initial conditions,
\begin{align}
\ddot{\bar{x}}-2\omega_{xy}\dot{\bar{y}}-(3\omega_{xy}^2+\alpha)\bar{x}-\beta\,\dot{\bar{y}}&=0,\notag\\
\ddot{\bar{y}}+2\omega_{xy}\dot{\bar{x}}&=-c_-F_{\mathrm{air}}(t),
\end{align}
where $\alpha=\tfrac{4\omega_{xy}^2}{c_-^2/s_{J_2}}$ and $\beta=\tfrac{4\omega_{xy}}{c_-^2/s_{J_2}}$. Taking Laplace transforms with $\bar{X}(s)=\mathcal{L}\{\bar{x}(t)\}$ and $\bar{Y}(s)=\mathcal{L}\{\bar{y}(t)\}$ yields
\begin{align}
\begin{bmatrix}
s^2-(3\omega_{xy}^2+\alpha) & -s(2\omega_{xy}+\beta)\\
2\omega_{xy}s & s^2
\end{bmatrix}
\begin{bmatrix}\bar{X}(s)\\ \bar{Y}(s)\end{bmatrix}
=\notag\\
\begin{bmatrix}0\\ -c_-F_{\mathrm{air}}(s)\end{bmatrix},
\label{eq:laplace_lin_sys_compact_app}
\end{align}
so that eliminating $\bar{X}(s)$ gives transfer functions $\bar{X}(s)=G_x(s)\,(-c_-F_{\mathrm{air}}(s))$ and $\bar{Y}(s)=G_y(s)\,(-c_-F_{\mathrm{air}}(s))$, where $G_x(s)$ and $G_y(s)$ are rational in $s$ with simple poles at $s=\pm j\omega_{xy}$ and a double pole at $s=0$. To express the closed-form responses concisely, introduce the gain $\Gamma$ and its square,
\begin{align}
\Gamma\coloneqq \frac{2c_+}{c_-\,\omega_{xy}},\qquad
\Gamma^2=\left(\frac{2c_+}{c_-\,\omega_{xy}}\right)^2,\qquad
\mathcal{K}_\epsilon\coloneqq \frac{\epsilon_2}{2}\Gamma.
\label{eq:Gamma_defs_app}
\end{align}
For the DC input $F_{\mathrm{air}}(t)=F_{\mathrm{DC}}$ (thus $F_{\mathrm{air}}(s)=F_{\mathrm{DC}}/s$), partial fractions of $G_{x,y}(s)(-c_-)F_{\mathrm{DC}}/s$ yield
\begin{align}
\bar{x}_{\mathrm{DC}}(t)&=\frac{F_{\mathrm{DC}}}{c_-}\,\Gamma\left(t-\frac{\sin(\omega_{xy}t)}{\omega_{xy}}\right),\notag\\
\bar{y}_{\mathrm{DC}}(t)&=\frac{F_{\mathrm{DC}}}{c_-}\left(\Gamma^2\bigl(1-\cos(\omega_{xy}t)\bigr)-\mathcal{K}_\epsilon t^2\right).
\label{eq:dc_response_app}
\end{align}
For a single harmonic component $\cos(\nu t+\psi)$, the closed-form response depends on whether the input frequency coincides with the natural in-plane frequency, so the solution is evaluated by the following branch: for the nonresonant case $\nu\neq\omega_{xy}$, the response contains bounded oscillatory terms and additional non-oscillatory terms; for the resonant case $\nu=\omega_{xy}$, the response contains terms that grow in time and must be treated separately. Using the notation $\omega\coloneqq\omega_{xy}$, $\Delta\coloneqq \omega^2-\nu^2$, $c_\psi\coloneqq\cos\psi$, $s_\psi\coloneqq\sin\psi$, $c_\nu(t)\coloneqq\cos(\nu t)$, $s_\nu(t)\coloneqq\sin(\nu t)$, $c(t)\coloneqq\cos(\omega t)$, and $s(t)\coloneqq\sin(\omega t)$, the nonresonant response for $\nu\neq\omega$ is
\begin{align}
\bar{x}_{\mathrm{harm}}(t)
&=\Gamma\Bigg[c_\psi\!\left(\frac{\omega^2}{\nu\Delta}s_\nu(t)-\frac{\omega}{\Delta}s(t)\right)
\notag\\&\qquad-s_\psi\!\left(\frac{1-c_\nu(t)}{\nu}+\frac{\nu}{\Delta}(c(t)-c_\nu(t))\right)\Bigg],\notag\\
\bar{y}_{\mathrm{harm}}(t)
&=c_\psi\!\left(\Gamma^2\frac{\omega^2}{\Delta}(c_\nu(t)-c(t))-\frac{\epsilon_2\Gamma}{2\nu^2}(1-c_\nu(t))\right)
\notag\\+&s_\psi\!\left(\Gamma^2\frac{\omega}{\Delta}(\nu s(t)-\omega s_\nu(t))+\frac{\epsilon_2\Gamma}{2\nu^2}(\nu t-s_\nu(t))\right),
\label{eq:harm_nonres_app}
\end{align}
whereas for the resonant case $\nu=\omega$ the response is
\begin{align}
\bar{x}_{\mathrm{res}}(t)
&=\Gamma\Bigg[c_\psi\!\left(\frac{s(t)}{2\omega}+\frac{t}{2}c(t)\right)
\notag\\&\qquad-s_\psi\!\left(\frac{1-c(t)}{\omega}-\frac{\omega t}{2}s(t)\right)\Bigg],\notag\\
\bar{y}_{\mathrm{res}}(t)
&=c_\psi\!\left(\Gamma^2\frac{\omega t}{2}s(t)-\frac{\epsilon_2\Gamma}{2\omega^2}(1-c(t))\right)\notag\\
&-s_\psi\!\left(\Gamma^2\frac{1}{2\omega}(s(t)-t c(t))-\frac{\epsilon_2\Gamma}{2\omega^2}(\omega t-s(t))\right).
\label{eq:harm_res_app}
\end{align}
In the present application, the along-track forcing is represented by the discrete set of even harmonics $\nu_m=4m\nu$, and the resonance condition is $\nu_m=\omega_{xy}$ for some integer $m$. When this equality does not hold, the particular solution for each harmonic is nonresonant and can be superposed, and the shift of the in-plane center over the uncontrolled interval is determined solely by the non-oscillatory terms of the nonresonant responses. Accordingly, the drag-induced center shift is obtained by extracting from the superposed particular solution only the terms free of $\sin(\cdot)$ and $\cos(\cdot)$, because only these terms modify the drift-center parameters. Writing the extracted non-circular drift in the transformed coordinates as $\bar{x}_{\mathrm{sp},o}(t)$ and $\bar{y}_{\mathrm{sp},o}(t)$ yields expressions of the form
\begin{align}
\bar{x}_{\mathrm{sp},o}(t)&=\frac{F_{\mathrm{DC}}}{c_-}\,\Gamma\,t + C_{1,\mathrm{air}},\\
\bar{y}_{\mathrm{sp},o}(t)&=\frac{F_{\mathrm{DC}}}{c_-}\,\bigl(\Gamma^2-\mathcal{K}_\epsilon t^2\bigr)+\frac{\epsilon_2}{2}C_{4,\mathrm{air}}-\frac{\epsilon_2}{2}C_{1,\mathrm{air}}t,
\label{eq:drift_compact_app}
\end{align}
and the original in-plane coordinates follow from $x_{\mathrm{sp},o}(t)=\bar{x}_{\mathrm{sp},o}(t)/c_+$ and $y_{\mathrm{sp},o}(t)=\bar{y}_{\mathrm{sp},o}(t)/c_-$. In the virtual-satellite formulation used to construct relative quantities, the DC component is common-mode and cancels; therefore, the center correction that enters the relative drift-center representation is defined by retaining only the nonresonant even-harmonic secular contributions summarized by $(C_{1,\mathrm{air}},C_{4,\mathrm{air}})$ in Eq.~\eqref{eq:C_air_defs_theta_xi_edge}, while the terms proportional to $F_{\mathrm{DC}}$ are excluded from the increment mapping. Denoting by $\bm r_o(t)$ the drift-center motion defined under the averaged-$J_2$ model and by $\bm r'_o(t)$ the corrected drift-center motion including aerodynamic drag, the corrected drift center is written by linear superposition as
\begin{align}
\label{eq:pf_ro_superposition_appA}
\bm r'_o(t)=\bm r_o(t)+\bm r_{\mathrm{air},o}(t),
\end{align}
where $\bm r_{\mathrm{air},o}(t)$ collects the constant term and the coefficient of $t$ induced by the nonresonant even harmonics and is parameterized by $(C_{1,\mathrm{air}},C_{4,\mathrm{air}})$ through the same drift-center structure. Matching the constant term and the drift term proportional to $t$ to the drift-center parametrization by the relative-element pair $(C_1,C_4)$ yields
\begin{align}
\label{eq:pf_ro_prime_appA}
\bm r'_o(t)
&=
\begin{bmatrix}
2C'_1\\[2pt]
C'_4-\epsilon_2 C'_1 t
\end{bmatrix}
\notag\\
&=
\begin{bmatrix}
2\!\left(C_1+\dfrac{1}{2}C_{1,\mathrm{air}}\right)\\[6pt]
C_4+\dfrac{\epsilon_2}{2}C_{4,\mathrm{air}}-\epsilon_2\!\left(C_1+\dfrac{1}{2}C_{1,\mathrm{air}}\right)t
\end{bmatrix},
\end{align}
so the injected mismatch model is extended by replacing $\bm r_{o,i}$ and $\bm r_{o,j}$ with $\bm r'_{o,i}$ and $\bm r'_{o,j}$. If $\nu_m=\omega_{xy}$ holds for some $m$, the corresponding resonant response must be used and the secular growth differs from the nonresonant case; the above mapping applies to the nonresonant regime where $\nu_m\neq\omega_{xy}$ for all contributing harmonics.

The Fourier series in Eq.~\eqref{eq:Fair_fourier_app} converges rapidly and can be truncated with high accuracy; the coefficient magnitude decays as $1/(16m^2-1)$, and a numerical evaluation gives
\begin{align}
\frac{1}{16m^2-1}\Big|_{m=1}^{5}&=\left\{\frac{1}{15},\frac{1}{63},\frac{1}{143},\frac{1}{255},\frac{1}{399}\right\}\notag\\\approx&\{0.0667,\,0.0159,\,0.0070,\,0.0039,\,0.0025\},
\label{eq:fourier_decay_app}
\end{align}
while Eq.~\eqref{eq:C_air_defs_theta_xi_edge} further includes $\nu_m=4m\nu$ and $\nu_m^2$ in the denominators, so the effective decay of the $m$th term is on the order of $1/m^3$ or faster, which justifies truncation at a small integer such as $m=5$ for implementation. Finally, the DC component is common-mode in the virtual-satellite formulation and cancels in relative quantities; therefore, the increments $(C_{1,\mathrm{air}},C_{4,\mathrm{air}})$ are formed from the nonresonant harmonic contributions only, and the DC terms are not included in these increments.

\subsection*{B. Formulation of Recursive Edge Propagation (Proof of Lemma~\ref{lem:expanding_graph_update})}\label{appendixB}
\begin{proof}\noindent
At stage $k+1$, the newly added edge set $\mathcal{F}_{k+1}$ is attached to a subset of nodes (anchor nodes) in the existing cluster.
The displacement of the anchor nodes over the interval $[t_k,t_{k+1})$ is therefore reflected in the relative drift-center coordinates of the new edges.

\medskip
\noindent{(1) Set of attachment nodes and the selection matrix $U_{k+1}$.}
Let $\mathcal U_{k+1}\subseteq \mathcal V_k$ be the set of existing nodes incident to at least one new edge, i.e.,
\begin{align*}
\mathcal U_{k+1}
&\coloneqq
\Bigl\{\,v\in\mathcal V_k \ \big|\ \exists f\in\mathcal F_{k+1}\ \text{s.t.}\ v\in\partial f \Bigr\},\\
u_{k+1}&\coloneqq |\mathcal U_{k+1}|.
\end{align*}
Let $\mathcal I_{k+1}=\{i_1<i_2<\cdots<i_{u_{k+1}}\}\subset\{1,\dots,n_k\}$ denote the ascending list of node indices corresponding to $\mathcal U_{k+1}$, and let $\bm e_i\in\mathbb R^{n_k}$ be the standard basis vector whose $i$-th entry is $1$.
Define the selection matrix
\begin{align}
\label{eq:U_def_rewrite_en}
U_{k+1}
\coloneqq
\begin{bmatrix}
\bm e_{i_1} & \bm e_{i_2} & \cdots & \bm e_{i_{u_{k+1}}}
\end{bmatrix}
\in\{0,1\}^{n_k\times u_{k+1}}.
\end{align}
Then, the stacked drift-center coordinates restricted to the attachment nodes are obtained as
\begin{align}
\label{eq:rU_def_rewrite_en}
\bm r_{o,\mathcal U}^{(k)}(t)
\coloneqq
\bigl(U_{k+1}^\top\otimes I_d\bigr)\bm r_o^{(k)}(t)
\in\mathbb R^{d u_{k+1}}.
\end{align}

\medskip
\noindent{(2) Anchor selection for each new edge and the connectivity matrix $V_{k,\mathcal F_{k+1}}$.}
Introduce a matrix $S_{k+1}\in\{0,1\}^{u_{k+1}\times|\mathcal F_{k+1}|}$ that specifies, for each new edge, which node in $\mathcal U_{k+1}$ serves as its anchor:
\begin{align}
\label{eq:S_def_rewrite_en}
S_{k+1}\in\{0,1\}^{u_{k+1}\times|\mathcal F_{k+1}|}.
\end{align}
Then, the connectivity matrix that maps the full node stack to the anchors of the new edges is expressed as
\begin{align}
\label{eq:V_def_rewrite_en}
V_{k,\mathcal F_{k+1}}
\coloneqq U_{k+1}S_{k+1}
\in\{0,1\}^{n_k\times|\mathcal F_{k+1}|}.
\end{align}

\medskip
\noindent{(3) Contribution of anchor displacement to the new-edge states.}
Define the contribution of the anchor displacement to the new-edge state stack by
\begin{align}
\label{eq:delta_rho_anc_def_rewrite_en_fixed}
&\delta\bm\rho_{\mathrm{anc}}
\coloneqq\notag\\
&\qquad\bigl(V_{k,\mathcal F_{k+1}}^\top\otimes I_d\bigr)
\Bigl(\bm r_o^{(k)}(t_{k+1}^-)-\bm r_o^{(k)}(t_k^+)\Bigr)
\in\mathbb R^{d|\mathcal F_{k+1}|}.
\end{align}

To express the node displacement in terms of the edge state, introduce the projector onto
$\mathrm{range}(E_k)$ by
\begin{align}
\label{eq:Pi_def_rewrite_en_fixed}
\Pi_k \coloneqq E_k L_{e,k}^\dagger E_k^\top \in \mathbb{R}^{n_k\times n_k}.
\end{align}
Using the node-to-edge map Eq.~\eqref{eq:stack_ro} and left-multiplying both sides by $(E_k L_{e,k}^\dagger\otimes I_d)$ yields
\begin{align}
\label{eq:node_recovery_projection_rewrite_en_fixed}
(\Pi_k\otimes I_d)\bm r_o^{(k)}(t)
=
(E_k L_{e,k}^\dagger\otimes I_d)\bm\rho^{(k)}(t).
\end{align}
Applying Eq.~\eqref{eq:node_recovery_projection_rewrite_en_fixed} to the displacement
$\bm r_o^{(k)}(t_{k+1}^-)-\bm r_o^{(k)}(t_k^+)$ gives
\begin{align}
\label{eq:node_disp_in_edge_rewrite_en_fixed}
&(\Pi_k\otimes I_d)\Bigl(\bm r_o^{(k)}(t_{k+1}^-)-\bm r_o^{(k)}(t_k^+)\Bigr)
=\notag\\
&\qquad\quad(E_k L_{e,k}^\dagger\otimes I_d)\Bigl(\bm\rho^{(k)}(t_{k+1}^-)-\bm\rho^{(k)}(t_k^+)\Bigr).
\end{align}

Since the existing-edge state contracts as
$\bm\rho^{(k)}(t_{k+1}^-)=\Phi_k\,\bm\rho^{(k)}(t_k^+)$, we have
\begin{align}
\bm\rho^{(k)}(t_{k+1}^-)-\bm\rho^{(k)}(t_k^+)
=
(\Phi_k-I_{dm_k})\,\bm\rho^{(k)}(t_k^+).
\end{align}
Substituting this into Eq.~\eqref{eq:node_disp_in_edge_rewrite_en_fixed} and then into
Eq.~\eqref{eq:delta_rho_anc_def_rewrite_en_fixed}, we obtain
\begin{align}
\label{eq:delta_rho_anc_proj_rewrite_en_fixed}
\delta\bm\rho_{\mathrm{anc}}
&=
\bigl(V_{k,\mathcal F_{k+1}}^\top\otimes I_d\bigr)
(\Pi_k\otimes I_d)\Bigl(\bm r_o^{(k)}(t_{k+1}^-)-\bm r_o^{(k)}(t_k^+)\Bigr)\nonumber\\
&=
\bigl(V_{k,\mathcal F_{k+1}}^\top E_k L_{e,k}^\dagger\otimes I_d\bigr)
(\Phi_k-I_{dm_k})\,\bm\rho^{(k)}(t_k^+).
\end{align}

\medskip
\noindent{(4) Construction of the new-edge states.}
The stacked relative drift-center coordinates on the new edges are given by the sum of the deployment input and the contribution induced by anchor displacement, shown as Figure~\ref{fig:new_edge_states}
\begin{align}
\label{eq:rho_new_total_rewrite_en}
\bm\rho_{\mathcal F_{k+1}}(t_{k+1}^-)
&=
-\delta\bm\rho_{\mathrm{anc}}+\bm w^{(k+1)}\nonumber\\
&=
R_{k,\mathcal F_{k+1}}
(I_{dm_k}-\Phi_k)\bm\rho^{(k)}(t_k^+)
+\bm w^{(k+1)}.
\end{align}
where, $R_{k,\mathcal F_{k+1}}$ follows the definition of Eq.~\eqref{eq:R_def_main}.
By assumption, the existing-edge states satisfy $\bm\rho^{(k)}(t_{k+1}^-)=\Phi_k\bm\rho^{(k)}(t_k^+)$.
Therefore, stacking the existing-edge and new-edge parts at time $t_{k+1}^+$ yields
\begin{align}
&\begin{bmatrix}
\bm\rho^{(k)}(t_{k+1}^-)\\
\bm\rho_{\mathcal F_{k+1}}(t_{k+1}^-)
\end{bmatrix}
=\notag \\ 
&\qquad\begin{bmatrix}
\Phi_k\\
\mathcal R_{k,\mathcal F_{k+1}}(I_{dm_k}-\Phi_k)
\end{bmatrix}\bm\rho^{(k)}(t_k^+)
+
\begin{bmatrix}
0\\
I_{d|\mathcal F_{k+1}|}
\end{bmatrix}\bm w^{(k+1)}.
\end{align}

\begin{figure}
    \centering
    \includegraphics[width=0.8\linewidth]{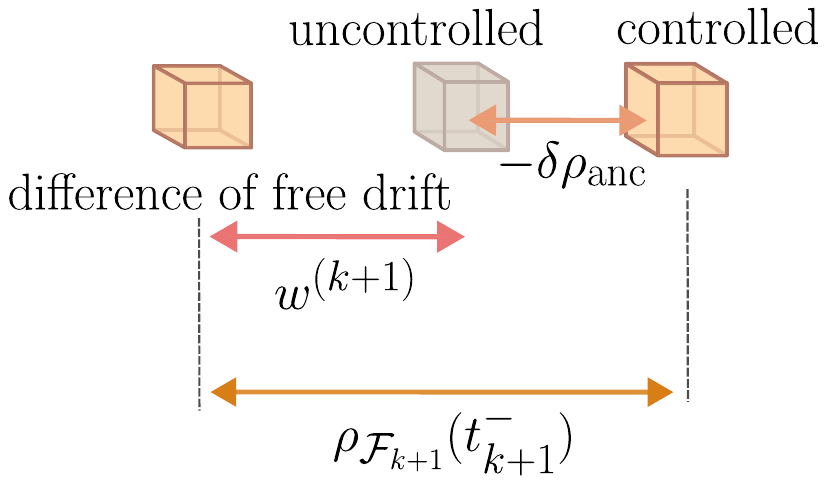}
    \caption{Image of the new-edge states}
    \label{fig:new_edge_states}
\end{figure}

Finally, left-multiplying by the permutation matrix $\mathcal P_{k+1}$ that aligns the edge ordering of $\mathcal G_{k+1}$ completes the derivation of
Eq.~\eqref{eq:rho_update_main}--\eqref{eq:Bk_def_main}.
\end{proof}

\subsection{C. Stage-wise Probaristic Safety Condition(Proof of Theorem~\ref{thm:stagewise_safety_main})}
\begin{proof}\noindent\label{appendixC}
Fix an arbitrary edge $e\in\mathcal{F}_{k+1}$.
Since the edge state satisfies
$\bm{\rho}_e\sim\mathcal{N}(\bm{\mu}_{k+1,e},\bm{\Sigma}_{k+1,e})$, we consider the $(1-\beta)$ confidence ellipsoid
\begin{align}
&\mathcal{E}_{1-\beta}
\coloneqq\notag\\&
\left\{
\bm{x}\in\mathbb{R}^d \,\middle|\,
(\bm{x}-\bm{\mu}_{k+1,e})^\top \bm{\Sigma}_{k+1,e}^{-1} (\bm{x}-\bm{\mu}_{k+1,e})
\le \chi^2_{d,\,1-\beta}
\right\}.
\end{align}
By the definition of the chi-square quantile, it holds that
$\mathbb{P}(\bm{\rho}_e\in\mathcal{E}_{1-\beta})=1-\beta$.

Moreover, 
it can be said that $\mathcal{E}_{1-\beta}$ is contained in the Euclidean ball centered at $\bm{\mu}_{k+1,e}$ with radius
\begin{align}
r_{e}\coloneqq \sqrt{\chi^2_{d,\,1-\beta}\;\lambda_{\max}(\bm{\Sigma}_{k+1,e})},
\end{align}
i.e.,
\begin{align}
\mathcal{E}_{1-\beta}\subseteq \left\{\bm{x}\in\mathbb{R}^d \,\middle|\, \|\bm{x}-\bm{\mu}_{k+1,e}\|_2\le r_e\right\}.
\end{align}
By the triangle inequality,
\begin{align}
\|\bm{x}\|_2 \le \|\bm{\mu}_{k+1,e}\|_2 + \|\bm{x}-\bm{\mu}_{k+1,e}\|_2,
\end{align}

the sufficient condition Eq.~\eqref{eq:sufficient_condition_main} implies
\begin{align}
\mathcal{E}_{1-\beta}\subseteq \left\{\bm{x}\in\mathbb{R}^d \,\middle|\, \|\bm{x}\|_2\le r_c\right\}.
\end{align}
Therefore,
\begin{align}
\mathbb{P}\!\left(\|\bm{\rho}_e\|_2\le r_c\right)
\;\ge\;
\mathbb{P}(\bm{\rho}_e\in\mathcal{E}_{1-\beta})
\;=\;1-\beta,
\end{align}
which is equivalent to $\mathbb{P}\!\left(\|\bm{\rho}_e\|_2>r_c\right)\le \beta$.
Since $e\in\mathcal{F}_{k+1}$ was arbitrary, Eq.~\eqref{eq:sufficient_condition_main} follows.
\end{proof}
\section*{Acknowledgements}
This work was supported by the JAXA Space Strategy Fund through the ”High accuracy satellite formation flight technology” (Grant No. JPJXSSF24MS09003).

The Gemini 2.5 Pro was used for assistance with language
editing, rephrasing, and translation in this manuscript.

\bibliographystyle{IEEEtran}
\bibliography{references.bib}

\thebiography
\begin{biographywithpic}
{Hideki Yoshikado}{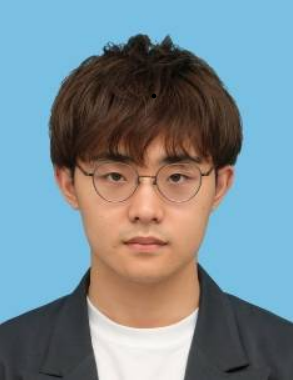} received the B.S. degree in Science and Engineering from the University of Tsukuba in 2024. He is currently pursuing the M.S. degree in the Department of Advanced Energy, Graduate School of Frontier Sciences, at The University of Tokyo, conducting his research at the Sakai Laboratory, JAXA. During his undergraduate studies, he served as the Project Manager for a CubeSat development project and was a visiting student at the University of California, Irvine, in 2023. Since 2024, he has been conducting research on control for electromagnetic formation flying in collaboration with Interstellar Technologies Inc. His research interests include proximity control for distributed space systems.
\end{biographywithpic}

\begin{biographywithpic}
{Yuta Takahashi}{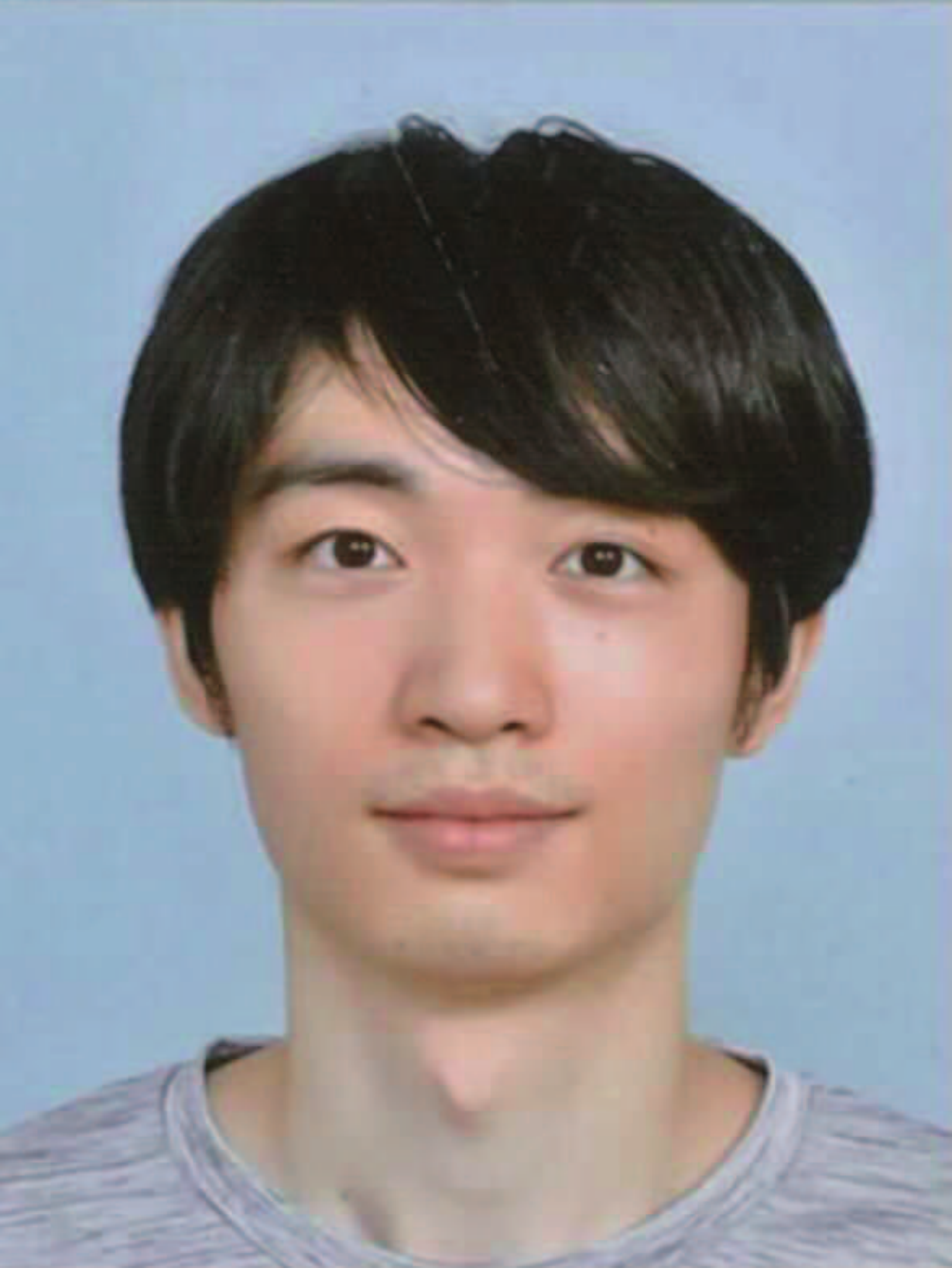} received the B.S. and M.S. degrees (Best Student Awards) in Mechanical Engineering (ME) in 2019 and 2021 (early graduations) from the Tokyo Institute of Technology, Japan. He was a visiting student researcher at the California Institute of Technology in 2022-2023.
He is currently pursuing a Ph.D. degree in ME at the Institute of Science Tokyo, Japan, under the JAXA/ISAS research program. His research fields include multi-agent and learning-based nonlinear control theory, with primary applications in distributed space systems and under-actuated systems. He led the avionics team for the space demonstration mission of the membrane array antenna from 2019 to 2023. Currently, he collaborates with Interstellar Technologies Inc., Japan, on the constellation project, utilizing a distributed space antenna from 2023. 

Mr. Takahashi was the recipient of Finalist award of the Best Student Paper Award of the GNC graduate student paper competition at the AIAA SciTech Forum 2025, the Best Paper Nominee as a Spotlight Speaker of the 2nd Space Robotics Workshop at IEEE SMC-IT/SCC 2025, the First Place Award at the MTT-Sat Challenge by the IEEE MTT-S in 2023, the Miura prize from The Japan Society of Mechanical Engineers in 2021, the Best Group Presentation Award of the AOTULE Conference, which he joined for his B.S. work in 2019, and the Best Student Award (2017, 2018, 2019-2021) from ME in Tokyo Tech.
\end{biographywithpic}


\end{document}